\documentclass[11pt]{article}
\usepackage[T1]{fontenc}
\usepackage[latin1]{inputenc}
\usepackage{geometry}
\usepackage{float}
\usepackage{mathrsfs}
\usepackage{amsmath}
\usepackage{amssymb}

\makeatletter

\usepackage{algorithm,algpseudocode}

\usepackage{amsthm}

\usepackage{mathrsfs}

\usepackage{amsfonts}

\usepackage{epsfig}

\usepackage{bm}

\usepackage{mathrsfs}

\usepackage{enumerate}

\numberwithin{equation}{section}

\@ifundefined{definecolor}{\@ifundefined{definecolor}
 {\@ifundefined{definecolor}
 {\usepackage{color}}{}
}{}
}{}
\usepackage{subfig}
\usepackage[all]{xy}

\usepackage{hyperref}
\hypersetup{
    colorlinks=true,
    linkcolor=blue,
    filecolor=magenta,      
    urlcolor=cyan,
    citecolor = red,
}
\makeatletter

\newtheorem{prop}{Proposition}[section]

\newcounter{hypA}
\newenvironment{hypA}{\refstepcounter{hypA}\begin{itemize}
  \item[({\bf A\arabic{hypA}})]}{\end{itemize}}
\newcounter{hypB}

\newcounter{hypD}

\newcommand{\ov}[1]{\overline{#1}}
\newcommand{\un}[1]{\underline{#1}}

\newcommand{\eqd}{\overset{d}{=}}

\usepackage{babel}\date{}

\usepackage{babel}

\makeatother

\usepackage{babel}

\newcommand{\bbN}{\mathbb{N}}

\newcommand{\bbP}{\mathbb{P}}

\newcommand{\R}{\mathbb{R}}

\begin{document}

\begin{center}

{\Large \textbf{Modeling of Measurement Error in Financial Returns Data}}

\vspace{0.5cm}

BY  AJAY JASRA$^{1}$, MOHAMED MAAMA$^{2}$ \& ALEKSANDAR MIJATOVI\'C$^{3}$

{\footnotesize $^{1}$School of Data Science,  The Chinese University of Hong Kong, Shenzhen, CN.
}\\
{\footnotesize $^{2}$Applied Mathematics and Computational Science Program,  Computer, Electrical and Mathematical Sciences and Engineering Division, King Abdullah University of Science and Technology, Thuwal, 23955-6900, KSA.}\\
{\footnotesize $^{3}$Department of Statistics, University of Warwick, UK.
}\\
{\footnotesize E-Mail:\,} \texttt{\emph{\footnotesize ajayjasra@cuhk.edu.cn, maama.mohamed@gmail.com,
a.mijatovic@warwick.ac.uk
}}

\end{center}

\begin{abstract}
In this paper we consider the modeling of measurement error for fund returns data.  In particular,  given access to
a time-series of discretely observed log-returns and the associated maximum over the observation period,  we develop a stochastic model which models the true log-returns and maximum via a L\'evy process and the data as a measurement error there-of.  The main technical difficulty of trying to infer this model,  for instance Bayesian parameter estimation,  is that the joint transition density of the return and maximum is seldom known, nor can it be simulated exactly.  Based upon the novel stick breaking representation of \cite{SBA} we provide an approximation of the model. We develop a Markov chain Monte Carlo (MCMC) algorithm to sample from the Bayesian posterior of the approximated posterior and then extend this to a multilevel MCMC method which can reduce the computational cost to approximate posterior expectations, relative to ordinary MCMC.  We implement our methodology on several applications including for real data.
\end{abstract}

\section{Introduction}

The modeling of financial data via L\'evy processes is ubiquitous in the literature of financial econometrics; see for instance \cite{bns, cont,gander,smc_levy} for several contributions.  Less common, however, is often the joint modeling of financial returns data along with their maximum over a given time period.  For instance,  that one has the index of a fund measured over a month and the associated maximum over that month.  The latter information can be of interest in financial applications, for instance in the case of prediction for risk assessment and so fourth. 
In this article we develop a statistical model for (not-necessarily) regularly observed financial data and the associated maximum,  modeled precisely as a \emph{latent} process.  The actual data are then assumed conditionally independent of all other variables, given the current value of the latent process and treated as a measurement error of such a process. It is well-known in the financial industry, especially in the context of hedge-fund data \cite{fung1,fung2}, that the afore mentioned data is subject to many biases and errors. To mitigate this issue, we assume that, in reality the data are drawn from an underlying process and that the data are noisy-observations of said process.  We then formulate a Bayesian model which allows one to infer unknown model parameters as well as to predict the data forward in time.

In the context of using L\'evy processes for such type of data,  one of the main barriers to the application of well-known statistical inferential procedures,  such as Bayesian estimation,  is the fact that the transition density of the value of its process,  along with it's maximum is seldom available in an analytic form.  Moreover, that direct simulation from such transition densities is not possible without some type of error,  for instance,  time discretization error when using time-stepping approximation.  For the model that we develop,  we use the 
novel stick breaking representation of \cite{SBA} (see also \cite{drawdown,cazares2021convex}) which provide a means to simulate the value of the process and it's maximum subject to an arbitrarily small error which decays geometrically in an accuracy parameter.  This latter property is key in producing Bayesian estimation that is as faithful as possible to the original model that we develop.

The main contributions of this paper are as follows:
\begin{itemize}
\item{We develop a new Bayesian model for discretely observed fund data with an associated maximum. This model is based upon a latent L\'evy process.}
\item{We provide an approximation of this model using the stick-breaking representation.}
\item{We provide a Markov chain Monte Carlo (MCMC) algorithm for sampling from the approximate Bayesian posterior.}
\item{We show how to utilize the multilevel Monte Carlo (MLMC) method to improve the cost to achieve a given
mean square error versus the afore-mentioned MCMC approach.}
\end{itemize}
The final contribution is worth elaborating upon here.  As previously mentioned,  our approximate model is based upon the stick-breaking representation of L\'evy process, its maxiumum and the duration as given in \cite{SBA}.
In particular, this representation is an infinite length sum of random variables that can be simulated (if the increments of the L\'evy process can be).  Of course,  in practice,  one cannot compute an infinite sum,  but truncate it at some point and this induces the approximation.  The idea of MLMC \cite{giles} is to then adopt a hierarchy of approximations that are increasingly accurate and write the expectation associated to the most precise approximation as a telescoping sum of expectations associated to increasingly less accurate approximations.
Then if one can appropriatelty simulate couplings of the probability distributions with `consecutive' levels of approximation it is possible to reduce the cost to approximate the expectation of interest up-to a given mean square error,  relative to using the most accurate approximation by itself; see \cite{giles1} for more details.
For the case of the model in this paper,  one has access to data,  so standard MLMC methodology must be
adapted as has been done in \cite{bayes_mlmc,ml_rev,ml_cont,med}.  Indeed we show how to combine the ideas of \cite{SBA} and \cite{bayes_mlmc}
to reduce the cost to achieve a given mean square error versus using the most precise approximation and MCMC.
This latter result is shown mathematically.

This paper is structured as follows.  
Section \ref{sec:model} presents the statistical model along with the approximation via the 
stick-breaking representation.  In Section \ref{sec:inference} we show how to fit the model to data using 
MCMC and multilevel MCMC methods.  This section also gives the theoretical result associated to the computational gains that are expected.  Finally in Section \ref{sec:numerics} we present several
numerical examples,  including an application to real data.

\section{Model and Approximation}\label{sec:model}

\subsection{Model}

We consider the regular in time observation of the log-returns and the maximal such value over several time instances.  We assume data observed over unit time intervals, but this is simply for notational convenience and are denoted as $D_T := \{(Y_1,\overline{Y}_1),\dots,  (Y_T,\overline{Y}_T)\}$, where the log-returns are the $Y_n$ and the associated maximum over the time instance $(n-1,n)$ is $\overline{Y}_n$, $n\in\{1,\dots,T\}$. 

More-formally, we consider a L\'evy process $\{X_t\}_{t\in[0,T]}$ that models the log-returns; such processes can capture a variety of different characteristics of the data and are thought as a realistic representation of financial data; see for instance \cite{cont}.  The constraint that we make is that one can sample the increments of the L\'evy process.  Un general this is not needed for modeling or for our methodology as one can use (approximations of)
the  L\'evy-Ito decomposition (e.g.~\cite{sato}) the description of our methods will be much simpler when this
is the case.
Now, writing the maximum of the process $\{X_t\}_{t\in[0,T]}$ on the interval $(n-1,n]$ as $\overline{X}_n$, we formulate a model for data as:
\begin{equation}
\label{eq:model_for_data}
(Y_n,\overline{Y}_n) | \{X_t\}_{t\in[0,T]},  D_T\setminus(Y_n,\overline{Y}_n) \sim 
\mathcal{N}_2\left((X_n,\overline{X}_n),\Sigma\right)
\end{equation}
where $\Sigma$ is a $2\times 2$ positive definite and symmetric matrix and $\mathcal{N}_2(\mu,\Sigma)$ is the bivariate Gaussian distribution with mean vector $\mu$ and covariance $\Sigma$.  We remark that Gaussian errors
are not necessary and indeed any (tractable) probability density could be used if needed.

We suppose that $\theta=(\Sigma,\phi)$ are unknown where $\phi$ are unknown (real-vector) parameters in the L\'evy process, 
 and with a prior density $\rho$ and our objective to infer the posterior
$$
\pi\left(d(\theta,x_{1:T},\overline{x}_{1:T})|D_T\right) \propto \left\{\prod_{n=1}^T g_{\theta}(z_n|(x_n,\overline{x}_n))
f_{\theta}(d(x_n,\overline{x}_n)|x_{n-1})\right\}\rho(\theta)d\theta
$$
where $z_n=(y_n,\overline{y}_n)$,  $g_{\theta}$ is the bivariate Gaussian density associated to a $\mathcal{N}_2\left((X_n,\overline{X}_n),\Sigma\right)$ distribution $f_{\theta}(d(x_n,\overline{x}_n)|x_{n-1})$ is the transition kernel associated to the L\'evy process, $x_0$ is taken as known and $d\theta$ is the appropriate dimensional Lebesgue measure.  In most applications of practical interest $f_{\theta}$ cannot be simulated without error and we introduce a particularly useful approximation in the next section.


\subsection{The Stick-Breaking Representation}
\label{subsec:SB-rep}

We consider a representation for  $\ov\chi_t:= (X_t,\ov{X}_t,\ov\tau_t(X))$ 
(resp. $\un\chi_t:= (X_t,\un{X}_t,\un\tau_t(X))$) for a given time horizon $t>0$.
Here $X_t$ is the L\'evy process,  $\ov{X}_t$ the maximum and $\ov\tau_t(X)$
the time when the maximum is attained. For all L\'evy processes considered in this paper, $\ov \tau_t$ is almost surely unique in the interval $(0,t)$.


Let $(U_m)_{m\in\bbN}$ be an i.i.d.~sequence with $U_m\sim\mathcal{U}_{(0,1)}$ $m\in\mathbb{N}$ with $\mathcal{U}_{(0,1)}$ denoting 
the uniform distribution on $(0,1)$.
A stick-breaking process 
$\ell=(\ell_m)_{m\in\bbN}$ on $[0,t]$  is given by
$L_0:= t$, $L_m:= L_{m-1}U_m$, 
$\ell_m:= L_{m-1}-L_m$ for $m\in\bbN$,
Let $Y$ be a L\'evy process,  independent of $\ell$, with the same law as $X$.
Then for \textit{any} L\'evy process $X$ we have~\cite[Theorem 11]{cazares2021convex}: 
$$\ov\chi_t\eqd	\sum_{k=1}^\infty\big(Y_{L_{k-1}}-Y_{L_k},
	\max\{Y_{L_{k-1}}-Y_{L_k},0\},
	\ell_k\cdot I_{\{Y_{L_{k-1}}-Y_{L_k}\geq 0\}}\big).$$
In particular, for any
$m\in\bbN$, the stick-breaking approximation (SBA) $\ov\chi_T^{(m)}$ is given as follows:
\begin{equation}\label{eq:chi}
\begin{split}
\ov\chi_t
&\eqd\big(Y_{L_m},\ov{Y}_{L_m},\ov\tau_{L_m}(Y)\big) + \ov\chi_T^{(m)},\qquad \text{where} \\
& \ov\chi_T^{(m)}:=
	\sum_{k=1}^m \big(Y_{L_{k-1}}-Y_{L_k},
	\max\{Y_{L_{k-1}}-Y_{L_k},0\},
	\ell_k\cdot I_{\{Y_{L_{k-1}}-Y_{L_k}\geq 0\}}\big).
\end{split}
\end{equation}
Note that, given $\ell$, the random vectors
$\ov\chi_T^{(m)}$ and 
$\big(Y_{L_m},\ov{Y}_{L_m},\ov\tau_{L_m}(Y)\big)$
in~\eqref{eq:chi} are independent. 

In many L\'evy models of interest it is possible to sample from the law $F(t,x)=\bbP(X_t\leq x)$, $x\in\R$, of  $X_t$ for any time horizon $t>0$ with constant complexity independent of $t$.
L\'evy models, such as CGMY (with stability parameter smaller than one), normal inverse Gaussian process (NIG), variance-gamma process (VG) are in this category.
 Then a procedure that simulates exactly from the law of the SBA 
$\ov\chi_T^{(m)}$ is given in Algorithm~\ref{alg:SBA}.

\begin{algorithm} 
	\caption{SB-Alg.}
	\label{alg:SBA}
	\begin{algorithmic}[1]
		\Require{$m\in\bbN$, fixed time horizon $t>0$} 
		\State{Set $\Lambda_0=t$, $\ov\chi_t^{(0)}=(0,0,0)$}
		\For{$k=1,\ldots,m$}
		\State{Sample $\upsilon_k\sim \mathcal{U}_{(0,1)}$ and put $\lambda_k=\upsilon_k\Lambda_{k-1}$ 
			and $\Lambda_k=\Lambda_{k-1}-\lambda_{k}$}
		\State{Sample $\xi_k\sim F(\lambda_k,\cdot)$ and put
			$\ov\chi_t^{(k)}= 
			\ov\chi_t^{(k-1)}+ 
			(\xi_k,\max\{\xi_k,0\},\lambda_k\cdot I_{\{\xi_k\geq 0\}})$} 
		\EndFor
		\State{Sample $\varsigma_m\sim F(\Lambda_m,\cdot)$ and \Return $\ov\chi_t^{(m)}+(\varsigma_m,\max\{\varsigma_m,0\},\Lambda_m\cdot I_{\{\varsigma_m\geq0\}})$}
	\end{algorithmic}
\end{algorithm}

Denote by $f_{\theta}^{(m)}(d(x_n,\overline{x}_n)|x_{n-1})$ the density of 
the output of Algorithm~\ref{alg:SBA}, using $m\in\bbN$
 sticks over the time horizon $t=1$,  shifted by $x_{n-1}$. 
An approximate posterior density in this case is given by 
\begin{equation}
\label{eq:approx_posterior_stick_breaking}
\pi^{(m)}\left(d(\theta,x_{1:T},\overline{x}_{1:T})|D_T\right) \propto \left\{\prod_{n=1}^T g_{\theta}(z_n|(x_n,\overline{x}_n))
f_{\theta}^{(m)}(d(x_n,\overline{x}_n)|x_{n-1})\right\}\rho(\theta)d\theta
\end{equation}
for any $m\in\bbN$.  We now focus on computational methodology to sample from the distribution associated to
\eqref{eq:approx_posterior_stick_breaking}.

\section{Inference}\label{sec:inference}

\subsection{Single Level MCMC}

We begin with a method to approximate expectations w.r.t.~\eqref{eq:approx_posterior_stick_breaking}. The structure of the model is that of a hidden Markov model and several approaches have been proposed in the literature. Perhaps the most popular is the particle MCMC method \cite{andrieu} which has been used in several related works \cite{bayes_mlmc,ml_cont,med}.  This is the approach that we adopt and it is given in Algorithm
\ref{alg:pmcmc}. Note that one needs to use Algorithm \ref{alg:pf} which calls Algorithm \ref{alg:SBA}. The context in which Algorithm \ref{alg:SBA} is used is of course that $t=1$ and we shift the process by its starting point as specified. Note that in Algorithm \ref{alg:pf} one typically chooses (as we do) $N=\mathcal{O}(T)$ but other choices have been investigated in the literature.

Under minimal assumptions Algorithm \ref{alg:pmcmc} provides a way to approximate expectations w.r.t.~$\pi^{(m)}$ in the following way.  For any $\varphi:\Theta\times\mathbb{R}^{2T}\rightarrow\mathbb{R}$ which is such
that
$$
\pi^{(m)}(\varphi) := \int_{\Theta\times\mathbb{R}^{2T}}\varphi(\theta,x_{1:T},\overline{x}_{1:T})\pi^{(m)}\left(d(\theta,x_{1:T},\overline{x}_{1:T})|D_T\right) 
$$
is finite
one has the estimate 
$$
\pi^{(m),S}(\varphi) := \frac{1}{S+1}\sum_{k=0}^S \varphi(\theta^k,x_{1:T}^k,\overline{x}_{1:T}^k).
$$
Moreover one has the almost sure convergence of $\pi^{(m),S}(\varphi)$ to $\pi^{(m)}(\varphi)$ as $S\rightarrow\infty$.

\begin{algorithm} 
\begin{algorithmic}[1]
		\Require{$N\in\mathbb{N}$ the number of particles,  $T$ the time horizon,  $m\in\mathbb{N}$ level of approximation,  $x_0$,  $D_T$ the data and $\theta$ the parameter.} 
\State{For $i=1,\dots,N$ sample $(X_1^i,\overline{X}_1^i)|x_0$ using Algorithm \ref{alg:SBA}.  Set $k=1$ 
$p_{\theta}^N(z_{-1})=1$
and go to step 2..}
\State{For $i=1,\dots,N$ compute $W_k^i=g_{\theta}(z_k|x_k^i,\overline{x}_k^i)/\{\sum_{j=1}^Ng_{\theta}(z_k|x_k^j,\overline{x}_k^j)\}$. 
Set $p_{\theta}^N(z_{1:k})=p_{\theta}^N(z_{1:k-1})\tfrac{1}{N}\sum_{ji=1}^Ng_{\theta}(z_k|x_k^i,\overline{x}_k^i)$
(where $p_{\theta}^N(z_{1:-1})=p_{\theta}^N(z_{-1})$).
Sample with replacement amongst the $(x_{1:k}^1,\overline{x}_{1:k}^1),\dots,(x_{1:k}^N,\overline{x}_{1:k}^N)$ using the weights $W_k^1,\dots,W_k^N$ calling the resulting samples
$(x_{1:k}^1,\overline{x}_{1:k}^1),\dots,(x_{1:k}^N,\overline{x}_{1:k}k^N)$ also. Go to step 3..
}
\State{For $i=1,\dots,N$ sample $(X_{k+1}^i,\overline{X}_{k+1}^i)|x_k^i$ using Algorithm \ref{alg:SBA}. 
Set $k=k+1$ and if $k=T$ go to step 4.~otherwise go to step 2.
}
\State{For $i=1,\dots,N$ compute $W_k^i=g_{\theta}(z_k|x_k^i,\overline{x}_k^i)/\{\sum_{j=1}^Ng_{\theta}(z_k|x_k^j,\overline{x}_k^j)\}$. 
Set $p_{\theta}^N(z_{1:k})=p_{\theta}^N(z_{1:k-1})\tfrac{1}{N}\sum_{ji=1}^Ng_{\theta}(z_k|x_k^i,\overline{x}_k^i)$.
Pick one trajectory $(x_{1:T}^1,\overline{x}_{1:T}^1),\dots,(x_{1:T}^N,\overline{x}_{1:T}^N)$
using the weights $W_T^1,\dots,W_T^N$.  Go to step 5..}
\State{\Return the selected trajectory $(x_{1:T},\overline{x}_{1:T})$ from step 4.  and $p_{\theta}^N(z_{1:T})$.}
\end{algorithmic}
\caption{Particle Filter.}\label{alg:pf}
\end{algorithm}

\begin{algorithm} 
\begin{algorithmic}[1]
		\Require{$N\in\mathbb{N}$ the number of particles,  $S$ number of MCMC samples, $q$ a positive Markov density on $\Theta$,  $T$ the time horizon,  $m\in\mathbb{N}$ level of approximation,  $x_0$,  $D_T$ the data and $\theta$ the parameter.} 
\State{Set $k=1$ and sample $\theta^0$ from $\rho$ and run Algorithm \ref{alg:pf} giving the initial
$(x_{1:T}^0,\overline{x}_{1:T}^0)$ and $p_{\theta^0}^{N,0}(z_{1:T})$. Go to step 2..}
\State{Propose $\theta'|\theta^{k-1}$ using the distribution associated to $q(\cdot|\theta^{k-1})$
and then run Algorithm \ref{alg:pf} with this given $\theta'$ yielding proposed
$(x_{1:T}',\overline{x}_{1:T}')$ and $p_{\theta'}^{N,'}(z_{1:T})$.  Compute
$$
A = \min\left\{1,
\frac{p_{\theta'}^{N,'}(z_{1:T})\rho(\theta')q(\theta^{k-1}|\theta')}{p_{\theta^{k-1}}^{N,k-1}(z_{1:T})\rho(\theta^{k-1})q(\theta'|\theta^{k-1})}
\right\}.
$$
Generate $U\sim\mathcal{U}_{(0,1)}$ and if $U<A$ set 
$(x_{1:T}^k,\overline{x}_{1:T}^k)=(x_{1:T}',\overline{x}_{1:T}')$, $
p_{\theta^k}^{N,k}(z_{1:T})=p_{\theta'}^{N,'}(z_{1:T})$
and $\theta^k=\theta'$. Otherwise set 
$(x_{1:T}^k,\overline{x}_{1:T}^k)=(x_{1:T}^{k-1},\overline{x}_{1:T}^{k-1})$,  $
p_{\theta^k}^{N,k}(z_{1:T})=p_{\theta^{k-1}}^{N,k-1}(z_{1:T})$
and $\theta^k=\theta^{k-1}$. Set $k=k+1$ and if $k=S+1$ go to step 3.  otherwise go to the start of step
2..
}
\State{\Return
$(\theta^0,x_{1:T}^0,\overline{x}_{1:T}^0),\dots,(\theta^S,x_{1:T}^S,\overline{x}_{1:T}^S)$.
}
\end{algorithmic}
\caption{Particle MCMC.}\label{alg:pmcmc}
\end{algorithm}

\subsection{Multilevel MCMC}

\subsubsection{Multilevel Identity}

Let $M\in\mathbb{N}$, $M>1$ be given,  then we have the trivial telescoping identity:
$$
\pi^{(M)}(\varphi) = \pi^{(1)}(\varphi) + \sum_{m=1}^{M-1} \left\{\pi^{(m+1)}(\varphi)-\pi^{(m)}(\varphi)\right\}
$$
assuming all expectations are well-defined,  which we assume throughout without mentioning further.
The basic idea of MLMC is then to approximate $\pi^{(m+1)}(\varphi)-\pi^{(m)}(\varphi)$,  for $m\in\{1,\dots,M-1\}$
using a dependent sampling strategy.  We will detail a high-level idea here and then subsequently provide specifics.

The scenario that we will now describe is not quite what we will implement, but, is close to the overall idea: the slight inconsisitencies are made to make the exposition at this stage easier to follow.
Let $\overline{\pi}^{(m+1)}$, $m\in\{1,\dots,M-1\}$ be a probability on $\mathsf{E}:=\Theta\times\mathbb{R}^{4T}$.
To simplify the notation set $u_n=(x_n,\overline{x}_n)$, $n\in\{1,\dots,T\}$ and as we will have another trajectory in $\mathsf{E}$, set $\overline{u}_n=(v_n,\overline{v}_n)\in\mathbb{R}^2$, $n\in\{1,\dots,T\}$.
Suppose that
each of $\pi^{(m)}$ and $\overline{\pi}^{(m+1)}$ have positive probability densities (denoted with the same symbol) w.r.t.~to some dominating $\sigma-$finite measure.  Then we can write:
\begin{equation}\label{eq:basic_id}
\pi^{(m+1)}(\varphi)-\pi^{(m)}(\varphi) = 
\overline{\pi}^{(m+1)}\left(\varphi_1 R_1\right) - \overline{\pi}^{(m+1)}\left(\varphi_2 R_2\right)
\end{equation}
where
\begin{eqnarray*}
\overline{\pi}^{(m+1)}\left(\varphi_1 R_1\right) & = & 
\int_{\mathsf{E}}
\varphi(\theta,u_{1:T}^{(m+1)})R_1(\theta,u_{1:T}^{(m+1)},\overline{u}_{1:T}^{(m+1)})
\overline{\pi}^{(m+1)}(\theta,u_{1:T}^{(m+1)},\overline{u}_{1:T}^{(m+1)})d(\theta,u_{1:T}^{(m+1)},\overline{u}_{1:T}^{(m+1)}) \\
\overline{\pi}^{(m+1)}\left(\varphi_2 R_2\right) & = & 
\int_{\mathsf{E}}
\varphi(\theta,\overline{u}_{1:T}^{(m+1)})R_2(\theta,u_{1:T}^{(m+1)},\overline{u}_{1:T}^{(m+1)})
\overline{\pi}^{(m+1)}(\theta,u_{1:T}^{(m+1)},\overline{u}_{1:T}^{(m+1)})d(\theta,u_{1:T}^{(m+1)},\overline{u}_{1:T}^{(m+1)})
\end{eqnarray*}
and
\begin{eqnarray*}
R_1(\theta,u_{1:T}^{(m+1)},\overline{u}_{1:T}^{(m+1)}) & = & \frac{\pi^{(m+1)}(\theta,u_{1:T}^{(m+1)})}
{\overline{\pi}^{(m+1)}(\theta,u_{1:T}^{(m+1)},\overline{u}_{1:T}^{(m+1)})} \\
R_2(\theta,u_{1:T}^{(m+1)},\overline{u}_{1:T}^{(m+1)}) & = & \frac{\pi^{(m)}(\theta,\overline{u}_{1:T}^{(m+1)})}
{\overline{\pi}^{(m+1)}(\theta,u_{1:T}^{(m+1)},\overline{u}_{1:T}^{(m+1)})}.
\end{eqnarray*}
The identity \eqref{eq:basic_id} successfully writes a difference of expectations as the expectation of a difference,
that is one need only sample a single probability to approximate the difference.  Moreover if:
$$
\overline{\pi}^{(m+1)}\left((\varphi_1 R_1-\varphi_2 R_2)^2\right)
$$
falls sufficiently fast as $m$ grows,  then it is possible that one can produce a method that approximates
\eqref{eq:basic_id} for each $m\in\{1,\dots,M-1\}$ independently,  that the cost to produce an approximation
(for a given error associated to a criterion)
of $\pi^{(M)}(\varphi)$ is reduced,  rather than just considering $\pi^{(M)}$ itself; see \cite{bayes_mlmc,ml_cont,med} for example.

The main issue is then how can one design $\overline{\pi}^{(m+1)}$ appropriately and this relies on simulating
couples from the SBA as in Algorithm \ref{alg:SBA}. This is the topic of the next section.

\subsubsection{Coupling}

The simulation algorithm for $(\ov\chi_t^{(m)}, \ov\chi_t^{(m+1)})$ is given in Algorithm \ref{alg:SBA_coupling}.
The approach is based on the obvious coupling $(\ov\chi_t^{(m)}, \ov\chi_t^{(m+1)})$ with the second component  obtained from the first by running Algorithm~\ref{alg:SBA} over the $(m+1)$-st stick, while keeping the other samples from  $\ov\chi_T^{(m)}$.
Note that the first and the second component of $(\ov\chi_t^{(m)},\ov\chi_t^{(m+1)})$ sampled by Algorithm~\ref{alg:SBA_coupling}  have the same law as the outputs of Algorithm~\ref{alg:SBA} for parameters $m$ and $m+1$, respectively. Note also that the final increments $\varsigma_m$ and $\varsigma_{m+1}$ are independent, while the increments $\xi_k$, $k\in\{1,\ldots,m\}$, are common in both components.

%

\begin{algorithm} 
	\caption{(SB-AlgC) Simulation of the coupling $(\ov\chi_t^{(m)}, \ov\chi_t^{(m+1)})$.}
	\label{alg:SBA_coupling} 
	\begin{algorithmic}[1]
		\Require{$m\in\bbN$, fixed time horizon $t>0$} 
		\State{Set $\Lambda_0=t$, $\ov\chi_t^{(0)}=(0,0,0)$}
		\For{$k=1,\ldots,m+1$}
		\State{Sample $\upsilon_k\sim \mathcal{U}_{(0,1)}$ and put $\lambda_k=\upsilon_k\Lambda_{k-1}$ 
			and $\Lambda_k=\Lambda_{k-1}-\lambda_{k}$}
		\State{Sample $\xi_k\sim F(\lambda_k,\cdot)$ and put
			$\ov\chi_t^{(k)}= 
			\ov\chi_t^{(k-1)}+ 
			(\xi_k,\max\{\xi_k,0\},\lambda_k\cdot I_{\{\xi_k\geq 0\}})$} 
		\EndFor
		
		\State{Sample $\varsigma_m\sim F(\Lambda_m,\cdot)$ and update $\ov\chi_t^{(m)}=\ov\chi_t^{(m)}+(\varsigma_m,\max\{\varsigma_m,0\},\Lambda_m\cdot I_{\{\varsigma_m\geq0\}})$}
		\State{Sample $\varsigma_{m+1}\sim F(\Lambda_{m+1},\cdot)$ and update $\ov\chi_t^{(m+1)}=\ov\chi_t^{(m+1)}+(\varsigma_{m+1},\max\{\varsigma_{m+1},0\},\Lambda_{m+1}\cdot I_{\{\varsigma_{m+1}\geq0\}})$}
		\State{\Return $(\ov\chi_t^{(m)},\ov\chi_t^{(m+1)})$}
			\end{algorithmic}
\end{algorithm}

\subsubsection{Identity}

Now to design the actual $\overline{\pi}^{(m+1)}$ that we will use,  we denote by $\overline{f}^{(m+1)}(d(u_n,\overline{u}_n)|x_{n-1},v_{n-1})$ as the coupled simulation as described in Algorithm \ref{alg:SBA_coupling} with the appropriate modifications as was done for Algorithm \ref{alg:SBA} and $f(u_n|x_{n-1})$.  Then we define
$$
\overline{\pi}^{(m+1)}\left(d(\theta,u_{1:T},\overline{u}_{1:T})|D_T\right) \propto \left\{\prod_{n=1}^T 
\overline{G}_{\theta}(z_n,u_n,\overline{u}_n)
\overline{f}^{(m+1)}(d(u_n,\overline{u}_n)|x_{n-1},v_{n-1})\right\}\rho(\theta)d\theta
$$
where we will take $\overline{G}_{\theta}(z_n,u_n,\overline{u}_n)=\max\{g_{\theta}(z_n|u_n),g_{\theta}(z_n|\overline{u}_n)\}$, but other choices are possible; see \cite{delta} for example.
Note that $v_0=x_0$ that is,  the second L\'evy trajectory which is used to help approximate $\pi^{(m)}$
is started at the same point as the L\'evy process.
Now set
\begin{eqnarray*}
R_1(\theta,u_{1:T},\overline{u}_{1:T}) & := & \prod_{n=1}^T\frac{g_{\theta}(z_n|u_n)}{\overline{G}_{\theta}(z_n,u_n,\overline{u}_n)} \\
R_2(\theta,u_{1:T},\overline{u}_{1:T}) & := & \prod_{n=1}^T\frac{g_{\theta}(z_n|\overline{u}_n)}{\overline{G}_{\theta}(z_n,u_n,\overline{u}_n)}
\end{eqnarray*}
Then we have the identity
$$
\pi^{(m+1)}(\varphi)-\pi^{(m)}(\varphi) = \frac{\overline{\pi}^{(m+1)}(\varphi_1R_1)}{\overline{\pi}^{(m+1)}(R_1)} - 
\frac{\overline{\pi}^{(m+1)}(\varphi_2R_2)}{\overline{\pi}^{(m+1)}(R_2)}.
$$
The main issue is now to derive an MCMC method to sample from $\pi^{(m+1)}$ which is the topic of the next section.

\subsubsection{MCMC Method for $\overline{\pi}^{(m+1)}$}

The MCMC method to sample from $\overline{\pi}^{(m+1)}$ is presented in Algorithms \ref{alg:delta_pf}
and \ref{alg:pmcmc_coup}.  The structure of the simulation is much the same as for Algorithms \ref{alg:pf}
and \ref{alg:pmcmc} with modifications for the extended state-space of the new target $\overline{\pi}^{(m+1)}$.
Expectations w.r.t.~$\overline{\pi}^{(m+1)}$ can be estimated in the following manner.  We set
\begin{eqnarray*}
\overline{\pi}^{(m+1),S}(\varphi_1 R_1) & := & \frac{1}{S+1}\sum_{k=0}^S\varphi(\theta^k,u_{1:T}^k)R_1(\theta^k,u_{1:T}^k,\overline{u}_{1:T}^k) \\
\overline{\pi}^{(m+1),S}(\varphi_2 R_2) & := & \frac{1}{S+1}\sum_{k=0}^S\varphi(\theta^k,\overline{u}_{1:T}^k)R_2(\theta^k,u_{1:T}^k,\overline{u}_{1:T}^k).
\end{eqnarray*}

\begin{algorithm} 
\begin{algorithmic}[1]
		\Require{$N\in\mathbb{N}$ the number of particles,  $T$ the time horizon,  $m\in\mathbb{N}$ level of approximation,  $x_0$,  $D_T$ the data and $\theta$ the parameter.} 
\State{For $i=1,\dots,N$ sample $(U_1^i,\overline{U}_1^i)|x_0$ using Algorithm \ref{alg:SBA_coupling}.  Set $k=1$ 
$\overline{p}_{\theta}^N(z_{-1})=1$
and go to step 2..}
\State{For $i=1,\dots,N$ compute $\overline{W}_k^i=\overline{G}_{\theta}(z_k,xu_k^i,\overline{u}_k^i)/\{\sum_{j=1}^N\overline{G}_{\theta}(z_k,u_k^j,\overline{u}_k^j)\}$. 
Set $\overline{p}_{\theta}^N(z_{1:k})=\overline{p}_{\theta}^N(z_{1:k-1})\tfrac{1}{N}\sum_{ji=1}^N\overline{G}_{\theta}(z_k,u_k^i,\overline{u}_k^i)$
(where $\overline{p}_{\theta}^N(z_{1:-1})=\overline{p}_{\theta}^N(z_{-1})$).
Sample with replacement amongst the $(u_{1:k}^1,\overline{u}_{1:k}^1),\dots,(u_{1:k}^N,\overline{u}_{1:k}^N)$ using the weights $\overline{W}_k^1,\dots,\overline{W}_k^N$ calling the resulting samples
$(u_{1:k}^1,\overline{u}_{1:k}^1),\dots,(u_{1:k}^N,\overline{u}_{1:k}k^N)$ also.  Go to step 3..
}
\State{For $i=1,\dots,N$ sample $(U_{k+1}^i,\overline{U}_{k+1}^i)|x_k^i,v_k^i$ using Algorithm \ref{alg:SBA_coupling}. 
Set $k=k+1$ and if $k=T$ go to step 4.~otherwise go to step 2.
}
\State{For $i=1,\dots,N$ compute $\overline{W}_k^i=\overline{G}_{\theta}(z_k,u_k^i,\overline{u}_k^i)/\{\sum_{j=1}^N\overline{G}_{\theta}(z_k,u_k^j,\overline{u}_k^j)\}$. 
Set $\overline{p}_{\theta}^N(z_{1:k})=\overline{p}_{\theta}^N(z_{1:k-1})\tfrac{1}{N}\sum_{ji=1}^N\overline{G}_{\theta}(z_k,u_k^i,\overline{u}_k^i)$.
Pick one trajectory $(u_{1:T}^1,\overline{u}_{1:T}^1),\dots,(u_{1:T}^N,\overline{u}_{1:T}^N)$
using the weights $\overline{W}_T^1,\dots,\overline{W}_T^N$.  Go to step 5..}
\State{\Return the selected trajectory $(u_{1:T},\overline{u}_{1:T})$ from step 4.  and $\overline{p}_{\theta}^N(z_{1:T})$.}
\end{algorithmic}
\caption{Delta Particle Filter.}\label{alg:delta_pf}
\end{algorithm}

\begin{algorithm} 
\begin{algorithmic}[1]
		\Require{$N\in\mathbb{N}$ the number of particles,  $S$ number of MCMC samples, $q$ a positive Markov density on $\Theta$,  $T$ the time horizon,  $m\in\mathbb{N}$ level of approximation,  $x_0$,  $D_T$ the data and $\theta$ the parameter.} 
\State{Set $k=1$ and sample $\theta^0$ from $\rho$ and run Algorithm \ref{alg:delta_pf} giving the initial
$(u_{1:T}^0,\overline{u}_{1:T}^0)$ and $\overline{p}_{\theta^0}^{N,0}(z_{1:T})$. Go to step 2..}
\State{Propose $\theta'|\theta^{k-1}$ using the distribution associated to $q(\cdot|\theta^{k-1})$
and then run Algorithm \ref{alg:delta_pf} with this given $\theta'$ yielding proposed
$(u_{1:T}',\overline{u}_{1:T}')$ and $\overline{p}_{\theta'}^{N,'}(z_{1:T})$.  Compute
$$
A = \min\left\{1,
\frac{\overline{p}_{\theta'}^{N,'}(z_{1:T})\rho(\theta')q(\theta^{k-1}|\theta')}{\overline{p}_{\theta^{k-1}}^{N,k-1}(z_{1:T})\rho(\theta^{k-1})q(\theta'|\theta^{k-1})}
\right\}.
$$
Generate $U\sim\mathcal{U}_{(0,1)}$ and if $U<A$ set 
$(u_{1:T}^k,\overline{u}_{1:T}^k)=(u_{1:T}',\overline{u}_{1:T}')$, $
\overline{p}_{\theta^k}^{N,k}(z_{1:T})=\overline{p}_{\theta'}^{N,'}(z_{1:T})$
and $\theta^k=\theta'$. Otherwise set 
$(u_{1:T}^k,\overline{u}_{1:T}^k)=(u_{1:T}^{k-1},\overline{u}_{1:T}^{k-1})$,  $
\overline{p}_{\theta^k}^{N,k}(z_{1:T})=\overline{p}_{\theta^{k-1}}^{N,k-1}(z_{1:T})$
and $\theta^k=\theta^{k-1}$. Set $k=k+1$ and if $k=S+1$ go to step 3.  otherwise go to the start of step
2..
}
\State{\Return
$(\theta^0,u_{1:T}^0,\overline{u}_{1:T}^0),\dots,(\theta^S,u_{1:T}^S,\overline{u}_{1:T}^S)$.
}
\end{algorithmic}
\caption{Particle MCMC for $\overline{\pi}^{(m+1)}$.}\label{alg:pmcmc_coup}
\end{algorithm}

\subsubsection{Algorithm and Estimate}\label{sec:est}

The approach that we shall use is now the following.
\begin{itemize}
\item{Run Algorithm \ref{alg:pmcmc} ($m=1$) with $S_1$ samples.}
\item{Indendently of Algorithm \ref{alg:pmcmc} and independently for $m\in\{1,\dots,M-1\}$
run Algorithm \ref{alg:pmcmc_coup} with $S_{m+1}$ samples.}
\end{itemize}
Then one has the estimator
$$
\widehat{\pi}^{(M)}(\varphi) = \pi^{(m),S_1}(\varphi) + \sum_{m=1}^{M-1}\left\{
\frac{\overline{\pi}^{(m+1),S_{m+1}}(\varphi_1 R_1)}{\overline{\pi}^{(m+1),S_{m+1}}(R_1)} -
\frac{\overline{\pi}^{(m+1),S_{m+1}}(\varphi_2 R_2)}{\overline{\pi}^{(m+1),S_{m+1}}(R_2)}
\right\}.
$$
The question now is how to choose $M$ and $S_1,\dots,S_M$, which is the topic of the next section.

\subsection{Theoretical Result}

Set $\Delta_m=\int_{\Theta}\rho(\theta)\mathbb{E}[(Y_{L_m}-Y_{L_{m+1}})^2]d\theta$ and $\overline{\Delta}_M=
\max\{|\int_{\Theta}\rho(\theta)\mathbb{E}[Y_{L_M}]|d\theta,|\int_{\Theta}\rho(\theta)\mathbb{E}[\overline{Y}_{L_M}]d\theta|\}$. We have the following result where we note that the assumption is stated in the appendix.  Below $\mathcal{B}_b(\mathsf{E})$ are the bounded and measurable real-valued functions on $\mathsf{E}$ and $\|\varphi\|_{\infty}=\sup_{x\in\mathsf{E}}|\varphi(x)|$.
$\textrm{Lip}(\mathsf{E})$ are the globally Lipschitz,  real-valued functions on $\mathsf{E}$, specifically
 the Lipschitz constant is written $\|\varphi\|_{\textrm{Lip}}$ and for any $(x,y)\in\mathsf{E}$,
$|\varphi(x)-\varphi(y)|\leq\|\varphi\|_{\textrm{Lip}}\|x-y\|_2$ where $\|\cdot\|_2$ is the $L_2-$norm.
Set $\|\varphi\|=\max\{\|\varphi\|_{\infty},\|\varphi\|_{\textrm{Lip}}\}$.
In the result below $\mathbb{E}[\cdot]$ denotes the expectation associated to the law of the randomness generated by the procedure in Section \ref{sec:est}.

\begin{prop}\label{prop:only_prop}
Assume (A\ref{ass:1}).  Then there exists a $C<+\infty$ such that for any $(M,S_{1:M},\varphi)\in\mathbb{N}^{M+1}\times\mathcal{B}_b(\mathsf{E})\cap\textrm{\emph{Lip}}(\mathsf{E})$:
$$
\mathbb{E}\left[\left(\widehat{\pi}^{(M)}(\varphi)-\pi(\varphi)\right)^2\right] \leq
C\|\varphi\|\left(\frac{1}{S_1}+\sum_{m=1}^{M-1}\frac{\Delta_m}{S_{m+1}}+
\sum_{m=1}^{M-1}\sum_{q=m+1}^M\frac{\Delta_m^{1/2}}{S_{m+1}}\frac{\Delta_q^{1/2}}{S_{q+1}}
+
\overline{\Delta}_M^2
\right).
$$
\end{prop}

\begin{proof}
Follows from the theory in \cite{bayes_mlmc,ml_cont} and the SB representation.  As the proofs in \cite{bayes_mlmc,ml_cont} are essentially repeated,  we omit them.
\end{proof}

The implication of this result very much depends upon the properties of the increments of the L\'evy process.
For instance, in the case that one has a Brownian motion,  one can take $M=\mathcal{O}(|\log(\epsilon)|)$
for some $\epsilon>0$ and $S_m=\mathcal{O}(\epsilon^{-2}\Delta_m^{\alpha})$,  $\alpha\in(0,1/2)$ and obtain a bound on the mean square errror,  which is what is given in Proposition \ref{prop:only_prop}, that is $\mathcal{O}(\epsilon^2)$.  The cost to achieve this error is the optimal $\mathcal{O}(\epsilon^{-2})$. If one only considers a single level $M$ then it is easy to see that the cost to achieve this same error is $\mathcal{O}(\epsilon^{-2}\log(\epsilon)^2)$.

\section{Numerical Results}\label{sec:numerics}

In this section, we provide detailed numerical illustrations of the performance of our proposed multilevel particle Metropolis-Hastings (MLPMMH) algorithm, benchmarked against the standard particle Metropolis-Hastings (PMMH) approach. Through a series of simulations, we demonstrate the efficacy of MLPMMH in tackling two stochastic equation models driven by L\'evy processes, which are commonly used in quantitative finance, showcasing its advantages and benefits in estimation and inference. Specifically, we consider three illustrative examples that highlight the capabilities of MLPMMH in handling complex models and real-world data. Firstly, we apply our algorithms to a subordinated Brownian motion model with a Gamma process, using both synthetic observations and real data from the S\&P 500 index, with a focus on mean squared error (MSE) and computational cost. Secondly, we extend our analysis to a subordinated Brownian motion model with an Inverse Gamma process, incorporating noisy observations data with an unknown covariance matrix and using the Hedge return fund data to illustrate the Bayesian parameter estimation capabilities of our methodology. Through these examples, we provide a thorough evaluation of the MLPMMH algorithm, highlighting its potential to improve the efficiency and accuracy of Bayesian inference in complex stochastic models.

\subsection{Model Settings}

We now turn to two numerical examples related to L\'evy processes, which are essential in modeling various types of stochastic dynamics.
\subsubsection{Model 1: Brownian Motion with Gamma (BMG) Subordinator}

Our first L\'evy process model involves time-changing a standard Brownian motion using an independent Gamma subordinator. The model's dynamics are governed by the following equation:
\begin{align}
\label{eqBMG}
X_t  = b \; t + \sigma  W_{Z_t}. 
\end{align}

In this context, \( b \) represents a constant drift coefficient, \( t \) denotes time, \( \sigma \) signifies the volatility coefficient, and \( W_{Z_t} \) denotes a one-dimensional standard Brownian motion subordinated by an increasing Gamma process \( Z_t \). The process \( Z_t \) follows a Gamma distribution with parameters \( \alpha t \) and \( \beta \), where \( \alpha \) and \( \beta \) are fixed constants. Specifically, in our simulations, we use \( \alpha = 1.5 \) and \( \beta = 2.0 \).

The model (\ref{eqBMG}) combines a linear deterministic trend with a stochastic term that incorporates a time-change process. The drift term, $bt$, captures the deterministic trend, while the subordinated Brownian motion, $\sigma W_{Z_t}$, introduces randomness and volatility.

\subsubsection{Model 2: Brownian motion with Inverse Gamma (BMIG) Subordinator}
Our second L\'evy process model is similar to the first but with an Inverse Gamma subordinator. The model's dynamics are still governed by the same equation:

\begin{align}
\label{eq:BM}
X_t  = b \; t + \sigma  W_{Z_t}, 
\end{align}
where $b$, $t$, and $\sigma$ are as before, but now $Z_t$ is an Inverse Gamma process.

In both models to be considered, for $k\in\{1,\dots,T\}$, we select the observations as $(Y_k, \overline{Y}_k)$, which, conditional on the process ${X_t}, \; {t \in [0,T]}$, follow a bivariate normal distribution:
$$
(Y_k,\overline{Y}_k) | \{X_t\}_{t\in[0,T]} \;  \sim 
\mathcal{N}_2\left((X_k,\overline{X}_k),\Sigma\right),
$$
where $\Sigma$ is a $2 \times 2$ positive definite and symmetric matrix, and $\mathcal{N}_2(\mu, \Sigma)$ denotes the bivariate Gaussian distribution with mean vector $\mu$ and covariance matrix $\Sigma$.

In the BMG model, the parameters to be estimated are $b$ and $\sigma$, which are assigned independent Gamma priors. A data set with $T$ observations is simulated with $b = 1$ and $\sigma = 0.5$. Specifically, we adopt the following prior distributions: $b \sim \mathcal{G}a(a_b, b_b)$ and $\sigma \sim \mathcal{G}a(a_\sigma, b_\sigma)$. 
In the simulated data case, we generated $T=200$ observations from the dynamics model and assigned priors of $\mathcal{G}a(1,0.5)$ (Gamma distribution of shape 0.5 and scale 1) for $\sigma$ and $\mathcal{G}a(1,1)$ for $b$.
For the real data, consisting of weekly log-returns from the S\&P 500 index, obtained from Yahoo Finance, covering the period from January 1, 2021, to November 1, 2023, the priors were $\mathcal{G}a(0.1, 0.1)$ for both parameters. 

In the BMIG model, the parameters to be inferred are the elements of the covariance matrix $\Sigma$. The observation data to be assimilated are finance data, specifically weekly hedge return fund data. The prior distributions used in numerical simulations are the inverse Wishart distribution $ \Sigma \sim \mathcal{IW}(\Psi, \nu)$ where
    \begin{itemize}
        \item $\Psi = \begin{pmatrix} 1 & 0 \\ 0 & 1 \end{pmatrix}$ is the scale matrix (identity matrix),
        \item $\nu = 3$ is the degrees of freedom.
    \end{itemize}
    
The density function of the Inverse Wishart distribution is given by:
    $$
        p(\Sigma) = \frac{|\Psi|^{\nu/2}}{2^{\nu p/2} \Gamma_p(\nu/2)} |\Sigma|^{-(\nu + p + 1)/2} \exp\left(-\frac{1}{2} \text{tr}(\Psi \Sigma^{-1})\right),
    $$
where $\Gamma_p(\cdot)$ is the multivariate Gamma function.

\subsection{Simulation Settings}

The simulation experiment is structured as follows. Our multilevel method employs a range of stick numbers, specifically $m \in \{5, 10, 15, 20, 30 \}$. For each stick configuration $m$, we deploy a particle filter with $\mathcal{O}(T)$ particles in the particle MCMC algorithm. In multilevel Monte Carlo simulations, the number of samples at each level is determined through an empirical process. Initially, each level is independently simulated to estimate variances and computational complexities. These estimates are then used in an optimization process to minimize total computational cost while maintaining accuracy. The number of samples is empirically validated, matching the error estimates of multilevel and single-level PMCMC samplers through extensive simulations. In all scenarios, a fixed burn-in period of $10000$ iterations is applied. For the particle filters, resampling is performed adaptively.

\subsection{Simulation Results}

We begin by considering the simulated data.
In Figure \ref{fig:msecost1} we can observe some output from the single level MCMC algorithm run at level $m=30$. We see the state-estimates and the autocorrelation plots from the chain. These indicate rather good mixing for this example, although we note of course that the samples here are not corrected by importance sampling.

In Figure \ref{fig:msecost1} we can see the cost-MSE plots (based upon 50 repeats). They clearly indicate that the multilevel MCMC method has a lower cost to achieve a given MSE. In Table \ref{tab:res} 
we estimate the rates, that is, log cost against log MSE based upon Figure \ref{fig:msecost1}. This suggests that a single level method has a cost of $\mathcal{O}(\epsilon^{-2} |\log(\epsilon)|^2)$ to achieve an MSE of $\mathcal{O}(\epsilon^2)$ at least up-to log-factors. For the multilevel estimators, we obtain the optimal rates. This is because Table  \ref{tab:res} says that the cost is $\mathcal{O}(\epsilon^{-2})$ to achieve an MSE of $\mathcal{O}(\epsilon^2)$; again up-to logarithmic factors. 
Figure \ref{fig:msecost1} and the second row of Table \ref{tab:res}  confirm similar results for the case of simulated data, in this real data setting.

\begin{figure}[H]
\centering
\subfloat[]{\includegraphics[width=0.45\textwidth]{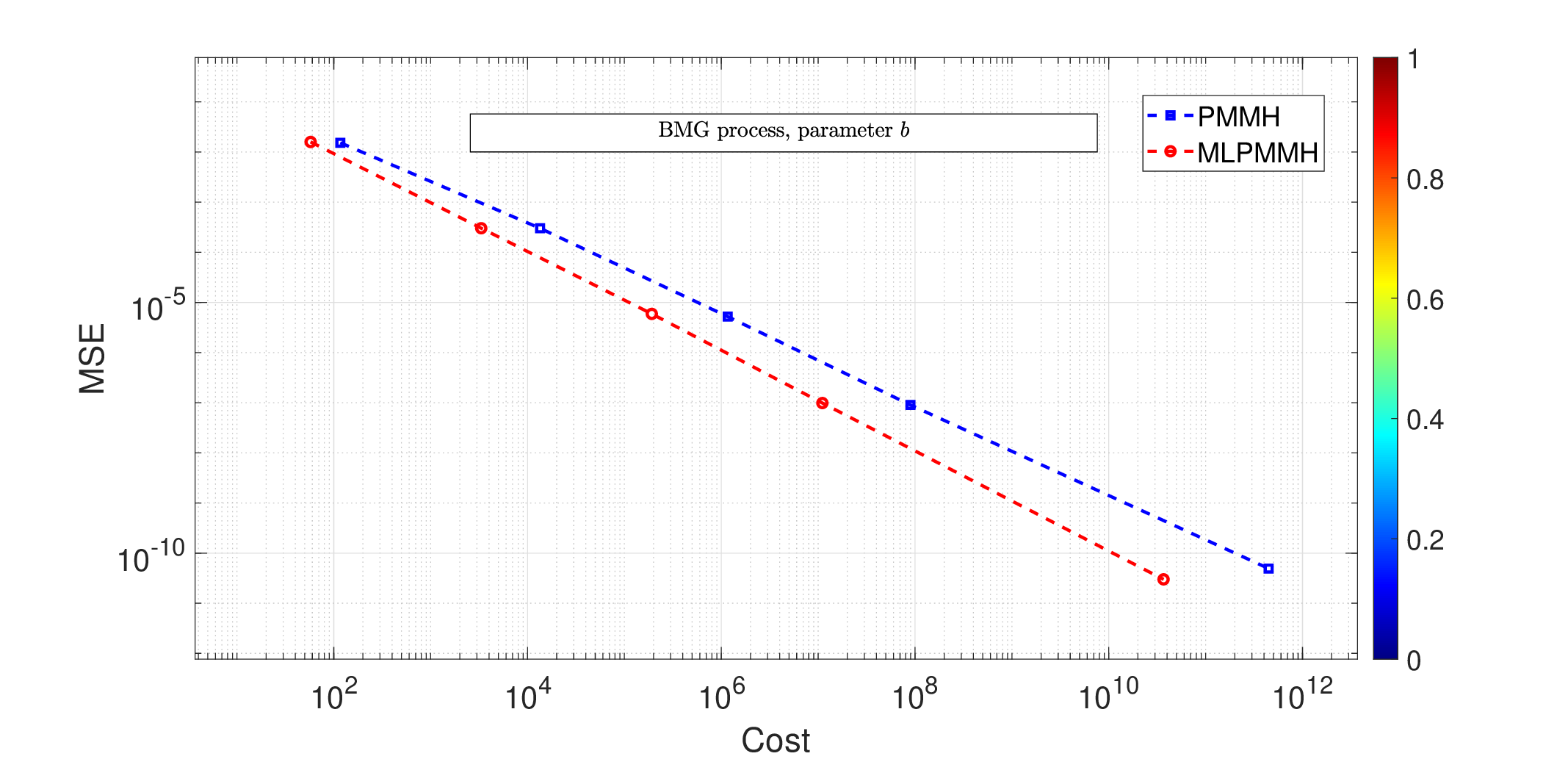}}\qquad
\subfloat[]{\includegraphics[width=0.45\textwidth]{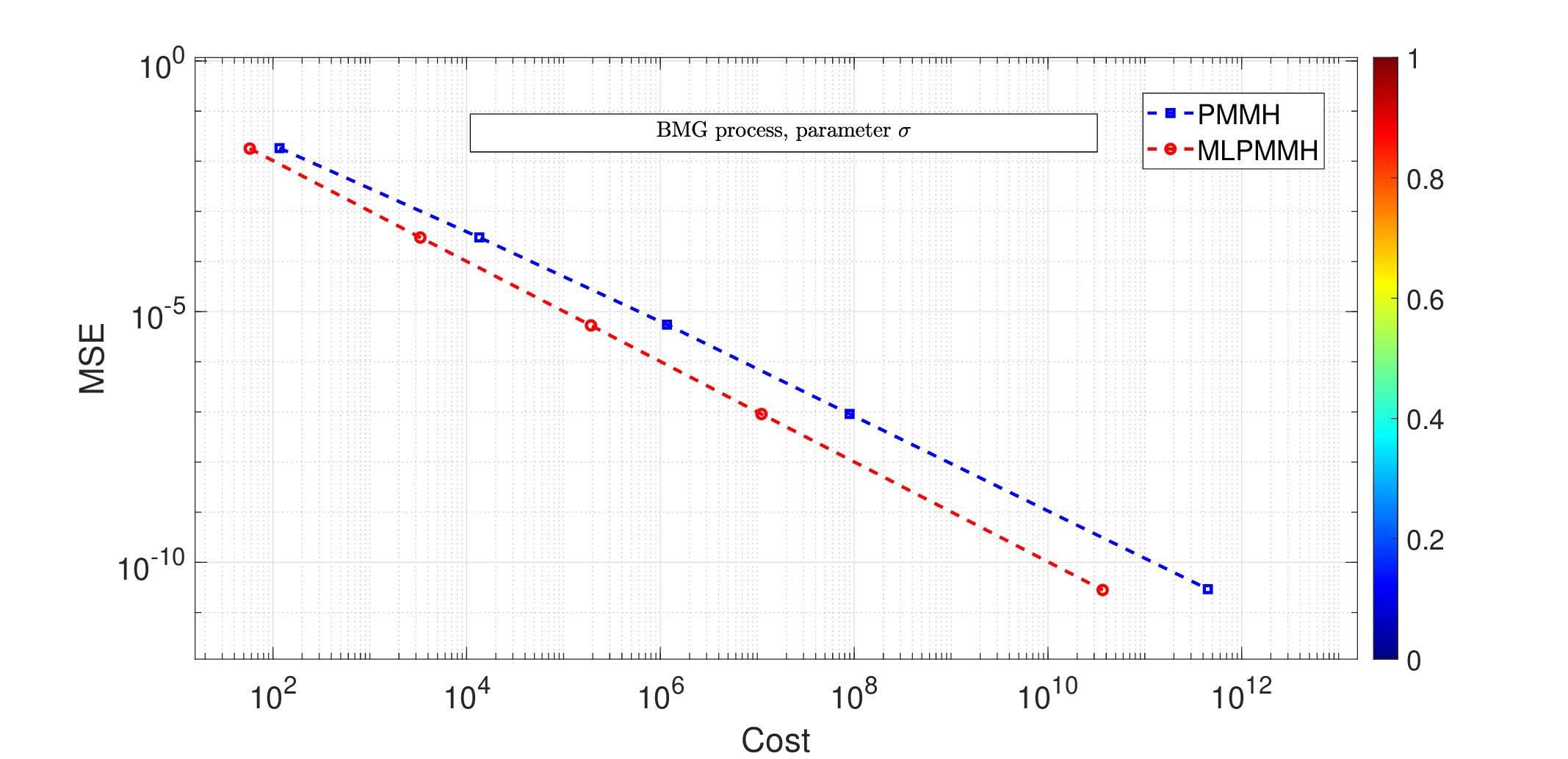}}\qquad
\subfloat[]{\includegraphics[width=0.45\textwidth]{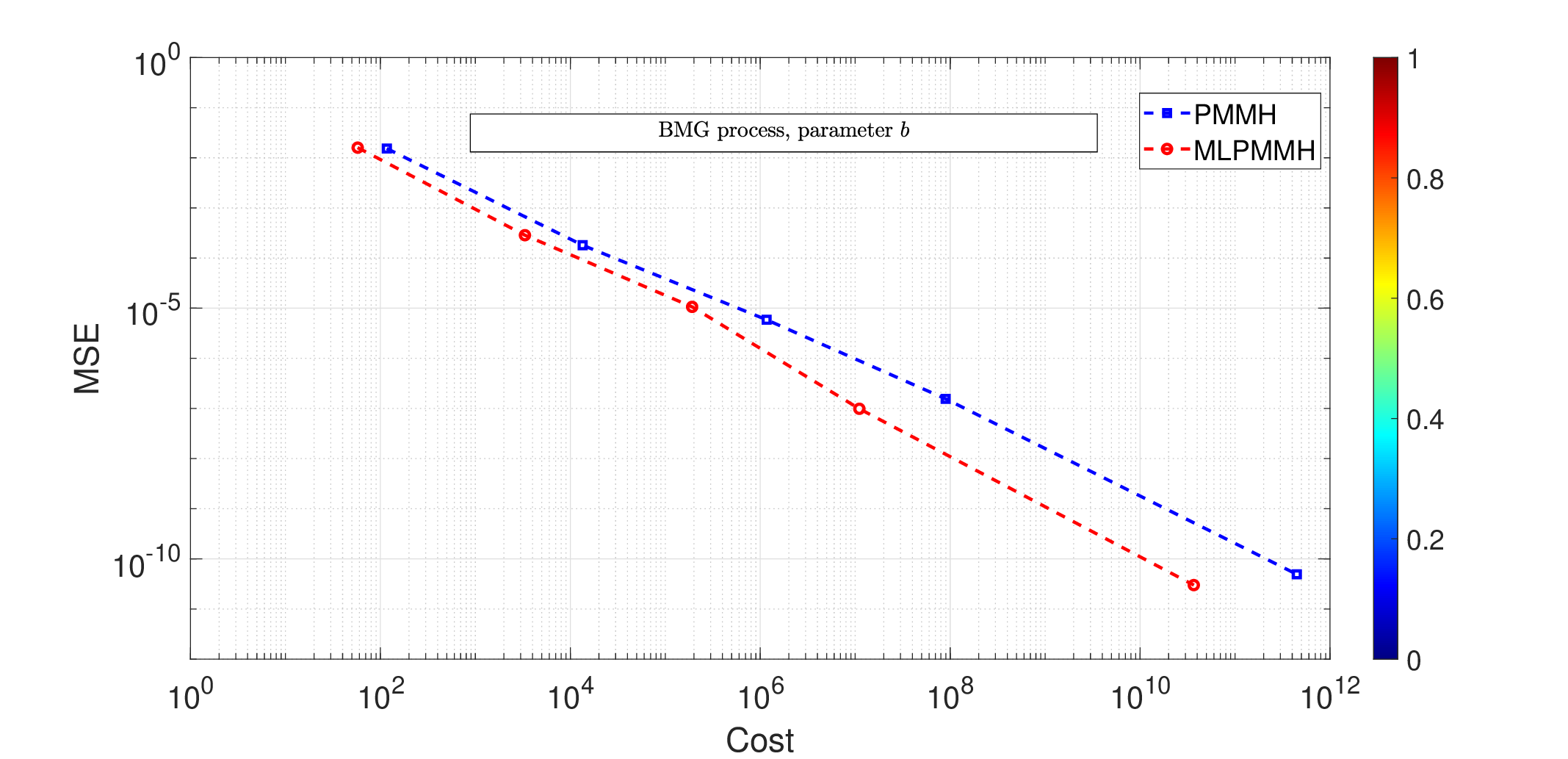}}\qquad
\subfloat[]{\includegraphics[width=0.45\textwidth]{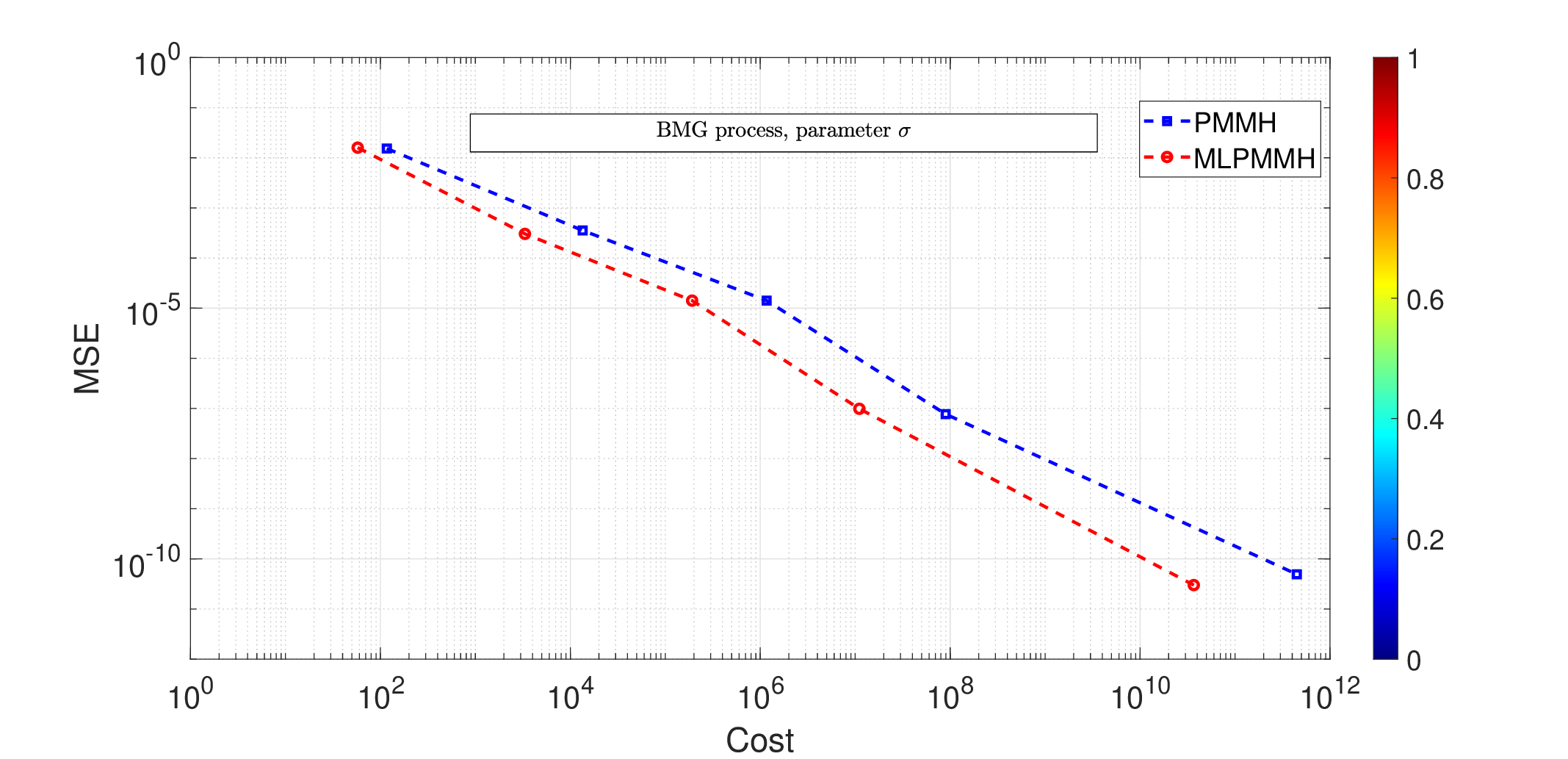}}
\caption{Comparison of cost versus mean squared error (MSE) plots for the BMG and BMIG model parameters using simulated and real data. (a) Parameter $b$ and (b) parameter $\sigma$ for simulated data with BMG process. (c) Parameter $b$ and (d) parameter $\sigma$ for real data with BMIG process.}
\label{fig:msecost1}
\end{figure}

\begin{table}[h]
\begin{center}
\begin{tabular}{ ||c|| c|| c|| c|| } 
\hline \hline
Data  & Parameter & PMCMC & MLPMCMC \\
\hline \hline
Synthetic & $b$ & -1.126 & -1.011 \\
 & $\sigma$ & -1.115 & -1.023 \\
 \hline
Real & $b$ & -1.131 & -1.018 \\
 & $\sigma$ & -1.112 & -1.021 \\
 \hline \hline
\end{tabular}
\caption{Estimated log cost versus the log of the MSE based upon the results from Figure \ref{fig:msecost1}.}
\label{tab:res}
\end{center}
\end{table}

We proceed by conducting experiments using real hedge fund return data, focusing on estimating the parameters of interest within the covariance matrix $\Sigma$. Figures \ref{fig:Inf1} - \ref{fig:Inf3}  show the performance of the PMCMC (again single-level PMCMC at level $m=30$) and is again very reasonable in terms of performance.

\begin{figure}[H]
\centering
\subfloat[]{\includegraphics[width=0.65\textwidth]{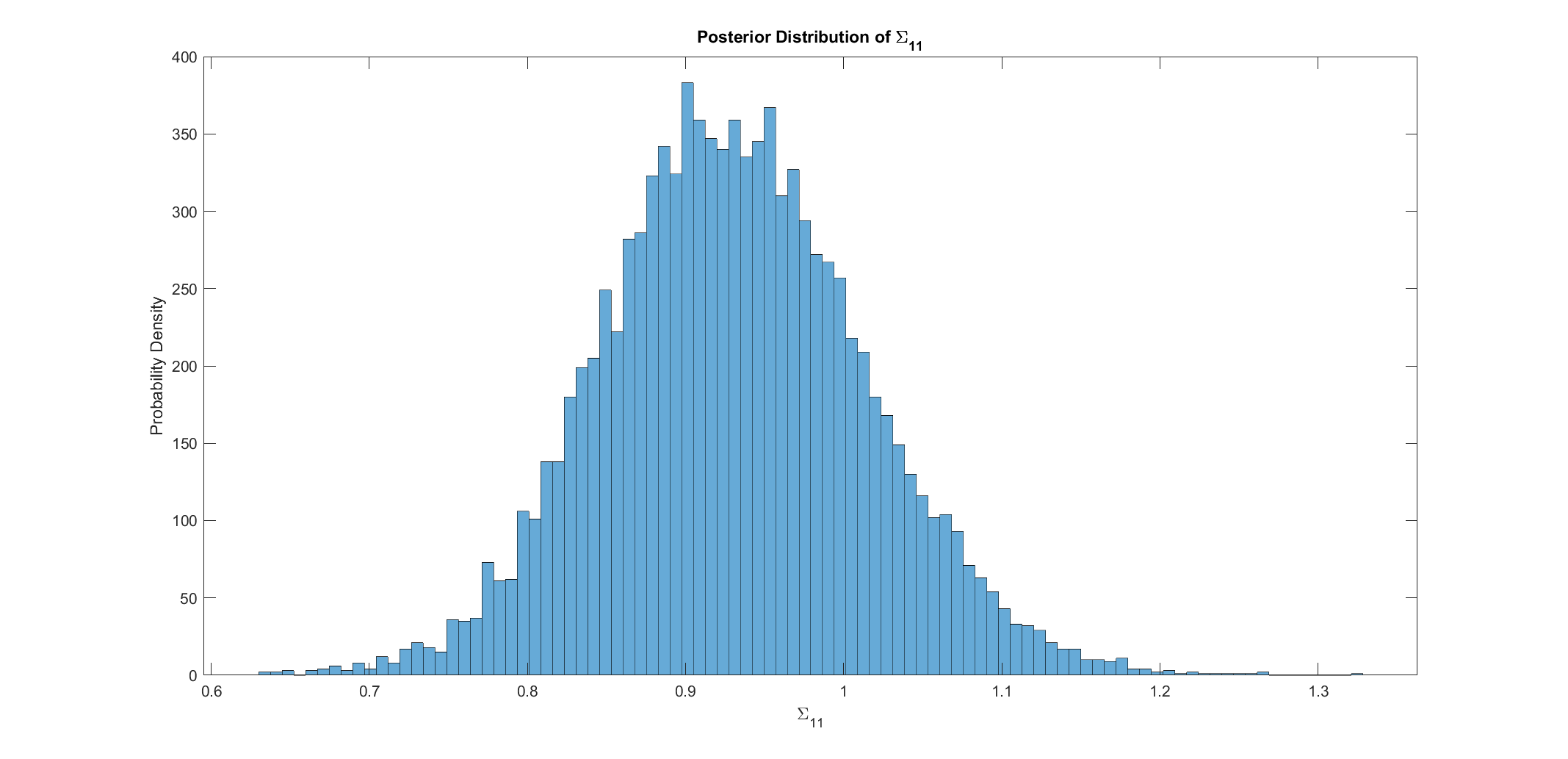}}\qquad
\subfloat[]{\includegraphics[width=0.65\textwidth]{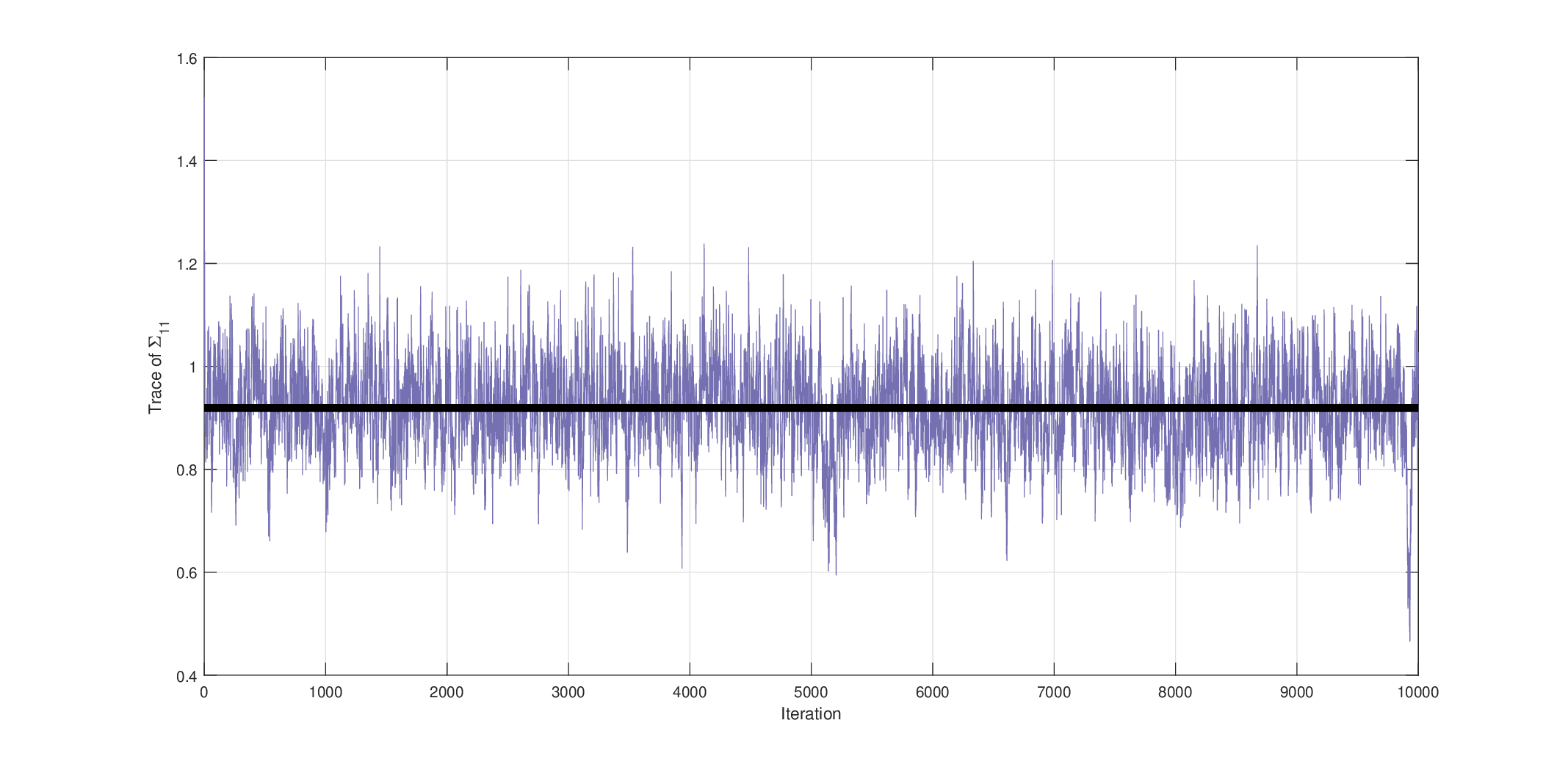}}\qquad
\subfloat[]{\includegraphics[width=0.65\textwidth]{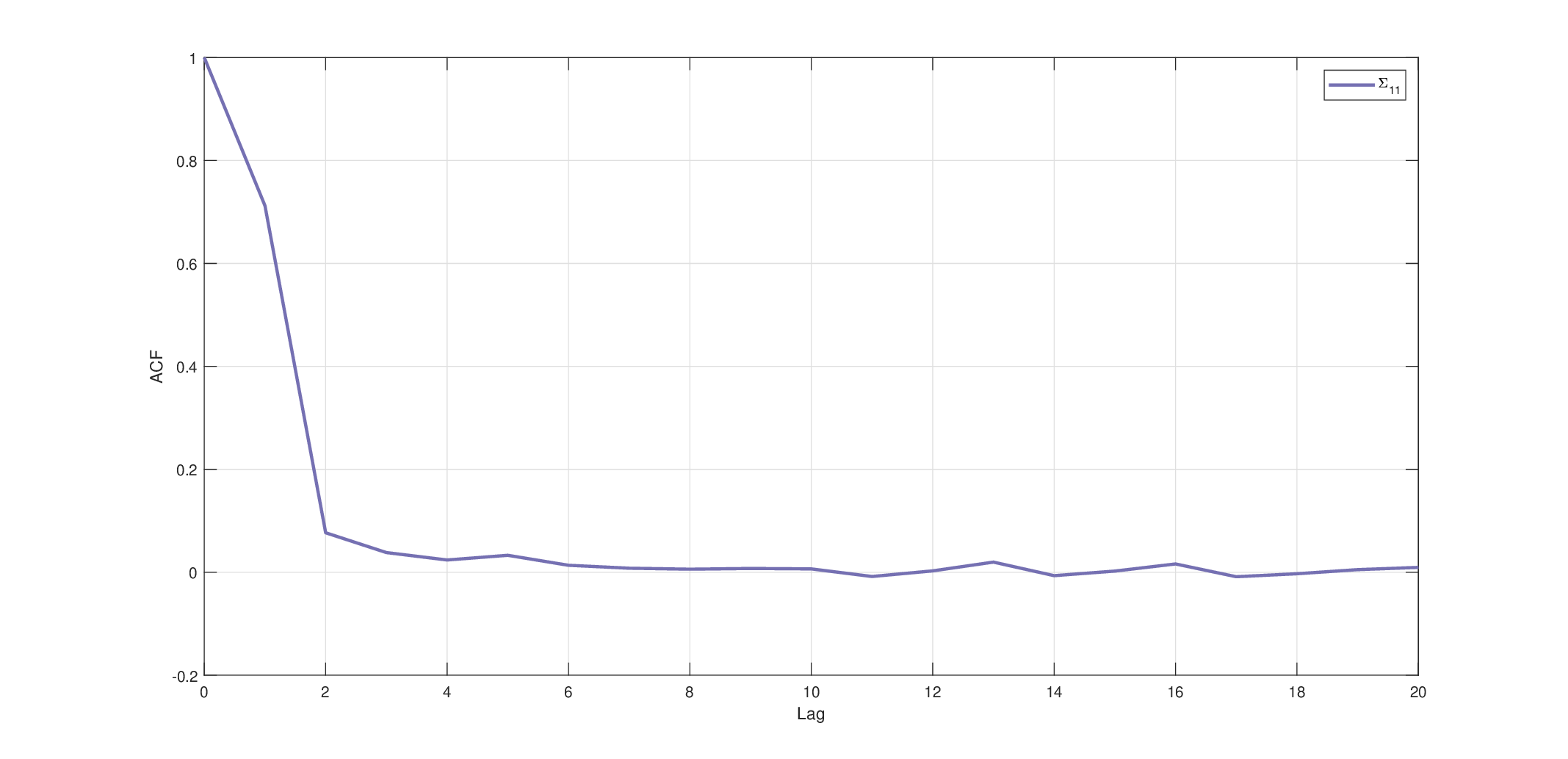}}\qquad
\caption{Bayesian Inference of Parameter $\Sigma_{11}$ in the BMIG Model with Hedge Fund Return Data. (a) The posterior distribution of the covariance matrix element $\Sigma_{11}$, derived from Bayesian Model Inference via hedge fund return data assimilation, illustrates the uncertainty and distributional characteristics of this parameter. (b) PMCMC trace plots for $\Sigma_{11}$ demonstrate the convergence behavior and sampling adequacy of the Markov Chain Monte Carlo algorithm, indicating the robustness of the posterior estimates. (c) Autocorrelation function of a representative PMCMC chain for $\Sigma_{11}$, revealing the degree of temporal correlation between samples and assessing the mixing efficiency of the chain.}
\label{fig:Inf1}
\end{figure}

\begin{figure}[H]
\centering
\subfloat[]{\includegraphics[width=0.75\textwidth]{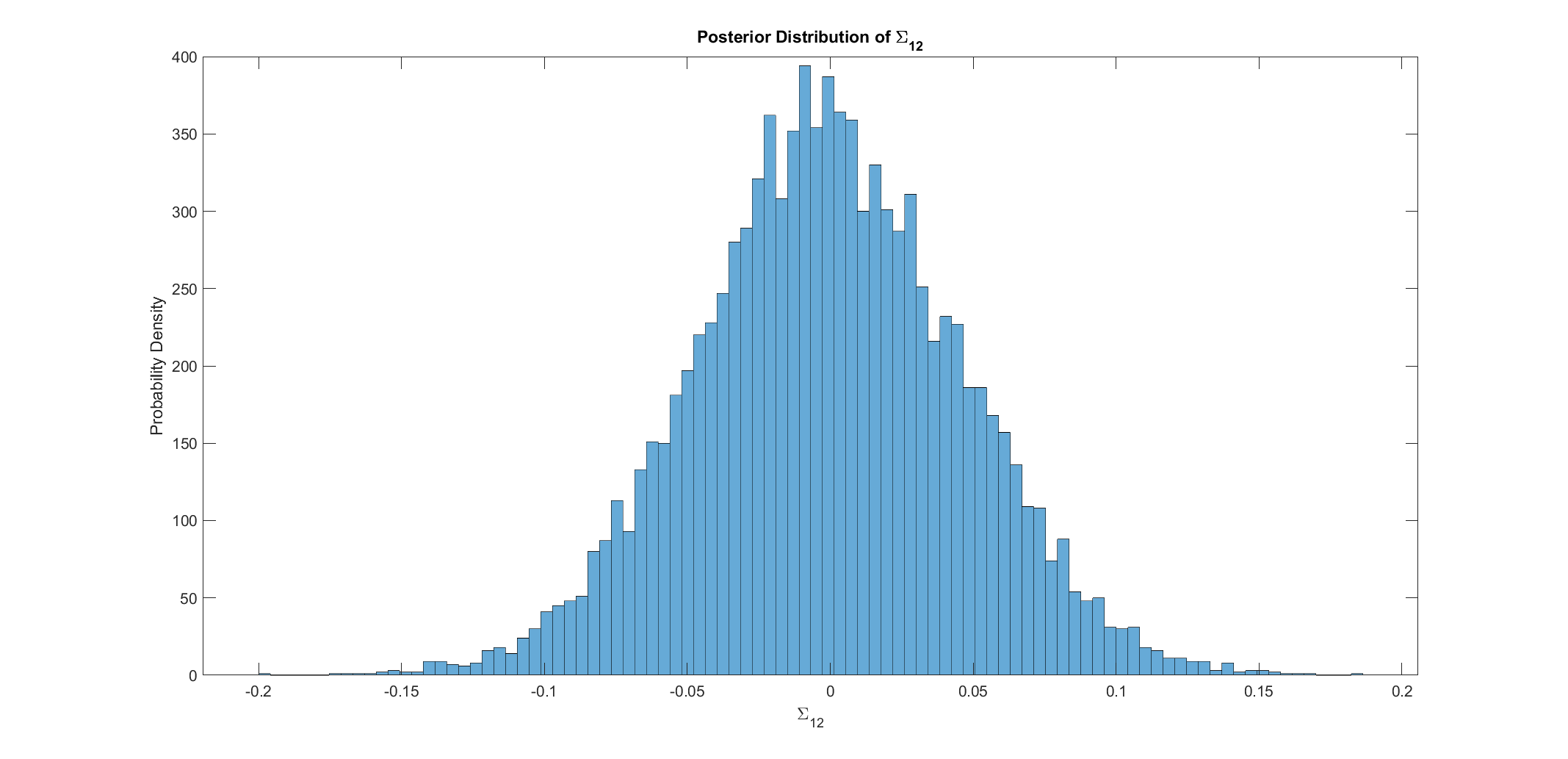}}\qquad
\subfloat[]{\includegraphics[width=0.75\textwidth]{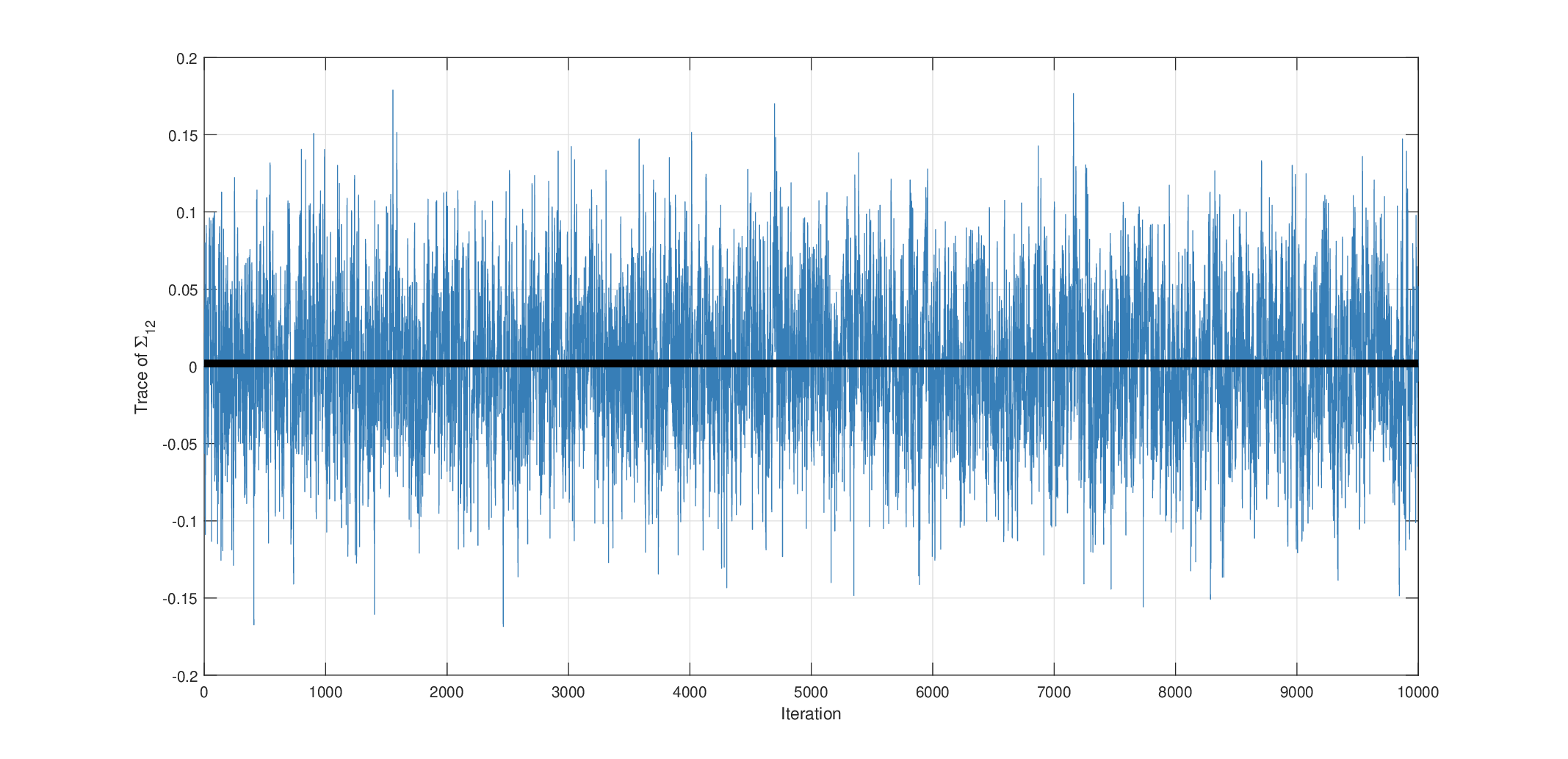}}\qquad
\subfloat[]{\includegraphics[width=0.75\textwidth]{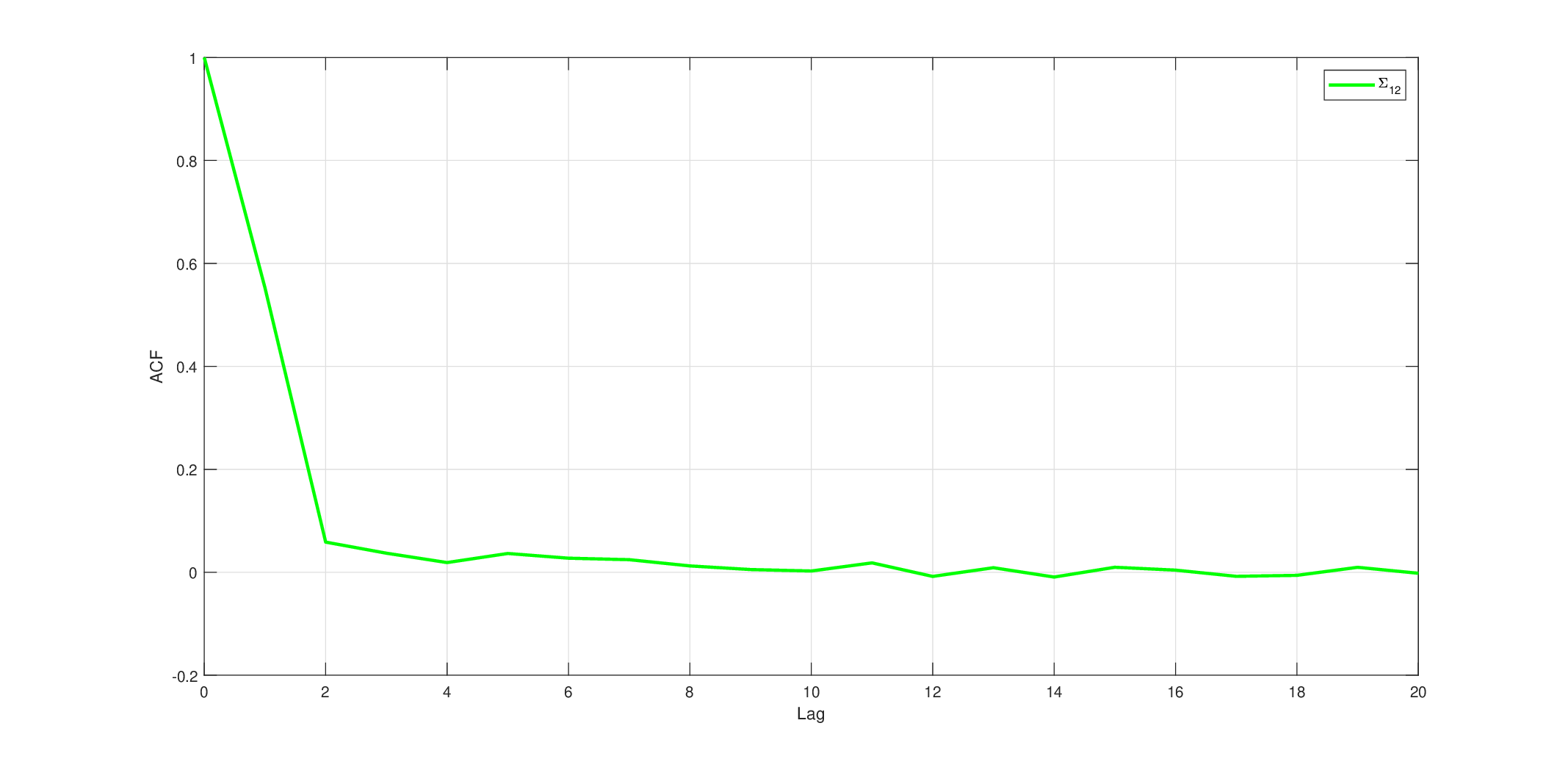}}\qquad
\caption{Bayesian Inference of Parameter $\Sigma_{12}$ in the BMIG Model with Hedge Fund Return Data. (a) The posterior distribution of the covariance matrix element $\Sigma_{12}$. (b) PMCMC trace plots for $\Sigma_{12}$. (c) Autocorrelation function of a representative PMCMC chain for $\Sigma_{12}$.}
\label{fig:Inf2}
\end{figure}

\begin{figure}[H]
\centering
\subfloat[]{\includegraphics[width=0.75\textwidth]{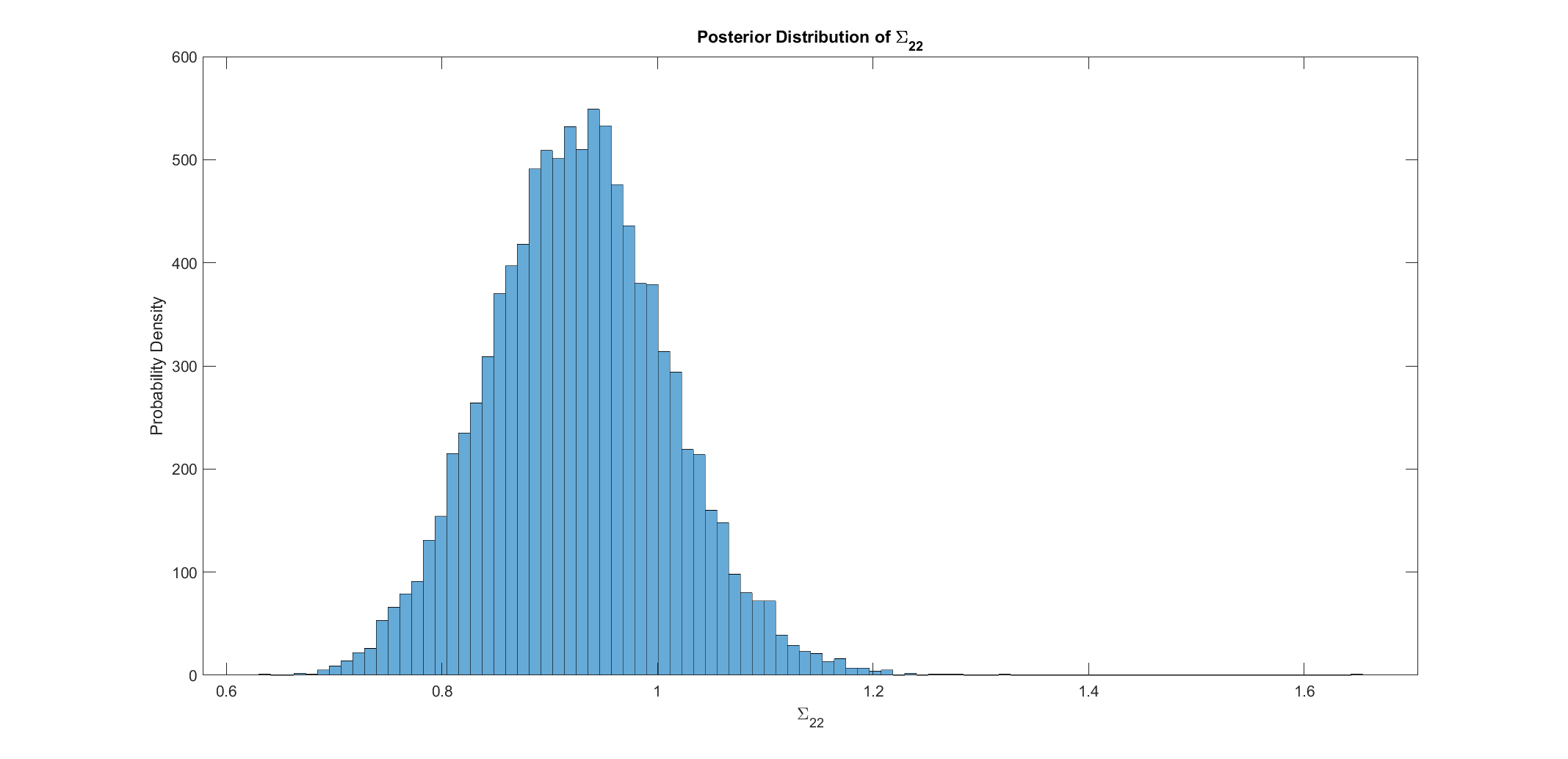}}\qquad
\subfloat[]{\includegraphics[width=0.75\textwidth]{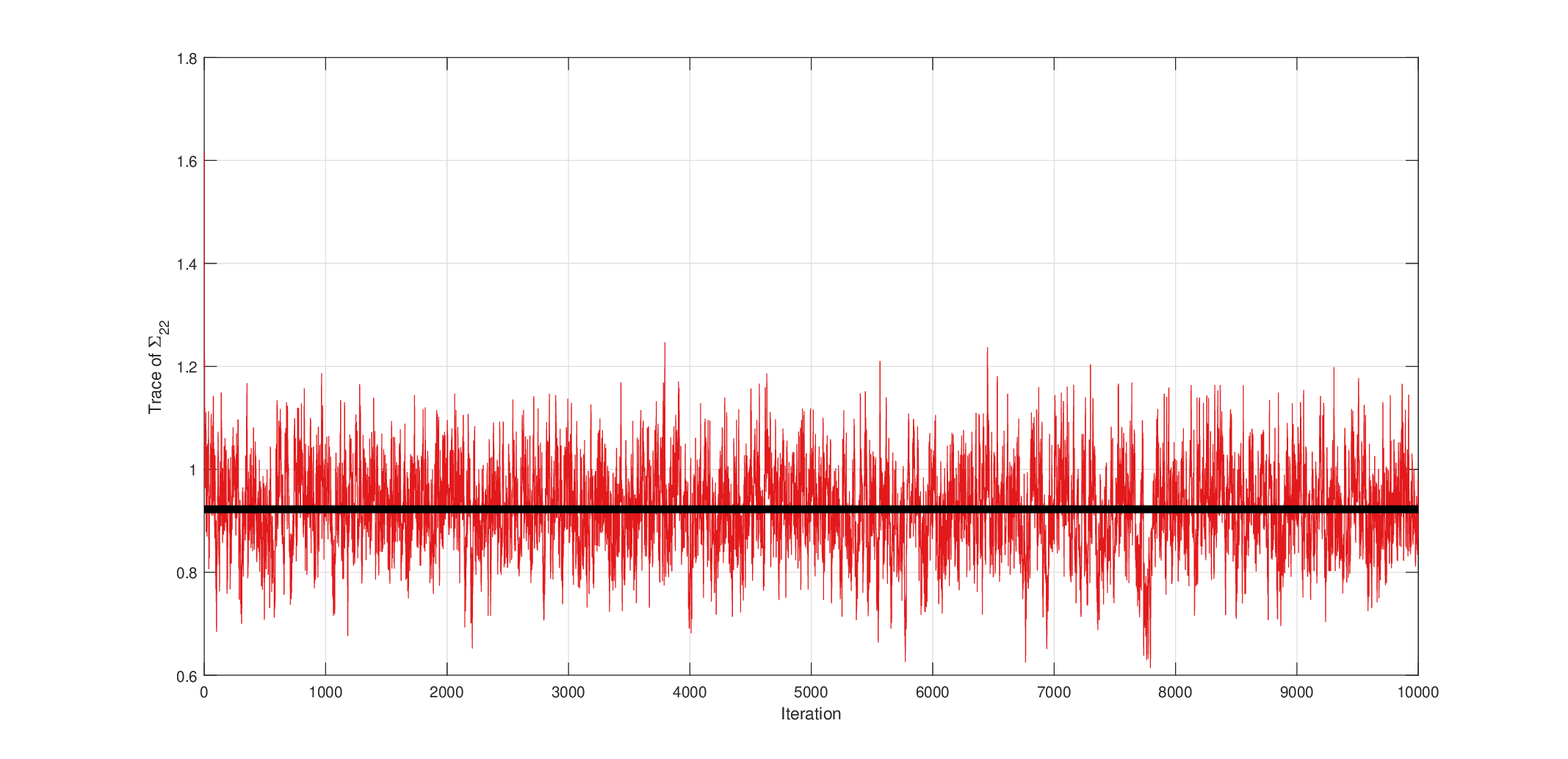}}\qquad
\subfloat[]{\includegraphics[width=0.75\textwidth]{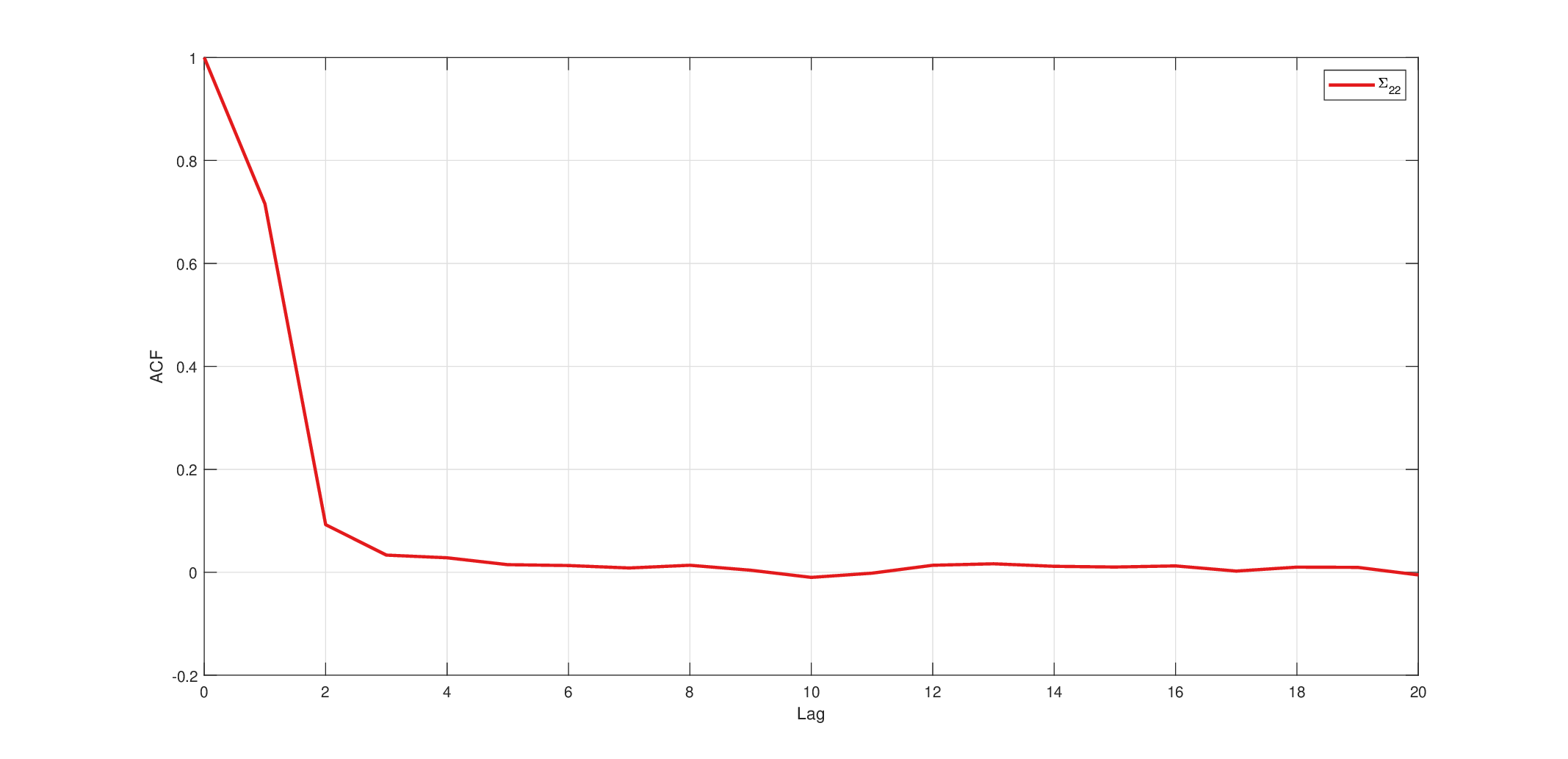}}\qquad
\caption{Bayesian Inference of Parameter $\Sigma_{22}$ in the BMIG Model with Hedge Fund Return Data. (a) The posterior distribution of the covariance matrix element $\Sigma_{22}$. (b) PMCMC trace plots for $\Sigma_{22}$. (c) Autocorrelation function of a representative PMCMC chain for $\Sigma_{22}$.}
\label{fig:Inf3}
\end{figure}

\appendix

\section{Assumption}

\begin{hypA}\label{ass:1}
\begin{enumerate}
\item{All the Markov chains are started in stationarity are reversible with respect to their invariant measure and
uniformly ergodic.}
\item{For any $z$ there exists $0<C_1<C_2<+\infty$ such that for any $u=(x,\overline{x}),\theta$
$C_1\leq g_{\theta}(z|u)\leq C_2$.  For every $z$ there exists a $C<+\infty$
such that for every $(\theta,\theta',u,u')$, we have $|g_{\theta}(z|u)-g_{\theta'}(z|u')|\leq C\|(\theta,u)-(\theta',u')\|_2$.
}
\end{enumerate}
\end{hypA}

This type of assumption has been discussed in detail in \cite{bayes_mlmc}.

\end{document}